\documentclass[conference,letterpaper]{IEEEtran}
\usepackage[utf8]{inputenc} 
\usepackage[T1]{fontenc}
\usepackage{url}
\usepackage{ifthen}
\usepackage{cite}
\usepackage[cmex10]{amsmath}
\usepackage{color}
 \usepackage{amsthm}
\usepackage{amsfonts}
\usepackage{enumitem}
\usepackage{multirow}
\usepackage[ruled,vlined]{algorithm2e}
\usepackage{tikz}
\usetikzlibrary{positioning}

\newtheorem{theorem}{Theorem}
\newtheorem{lemma}{Lemma}

\newtheorem{proposition}{Proposition}
\newtheorem{corollary}{Corollary}
\newtheorem{example}{Example}

\newcommand{\Hc}{\mathcal{H}}

\newcommand{\Wc}{\mathcal{W}}
\newcommand{\Xc}{\mathcal{X}}
\newcommand{\Yc}{\mathcal{Y}}
\newcommand{\Zc}{\mathcal{Z}}

\newcommand{\Eb}{\mathbb{E}}

\newcommand{\Rb}{\mathbb{R}}

\newcommand{\gen}{\mathrm{gen}}

\allowdisplaybreaks

\begin{document}
\title{Individually Conditional Individual Mutual Information Bound on Generalization Error} 

\author{
  \IEEEauthorblockN{Ruida Zhou, Chao Tian, and Tie Liu}
  \IEEEauthorblockA{Department of Electrical and Computer Engineering\\
  		   Texas A\&M University\\
                    Email: \{ruida, chao.tian, tieliu\}@tamu.edu}
}

\maketitle

\begin{abstract}
We propose a new information-theoretic bound on generalization error based on a combination of the error decomposition technique of Bu et al. and the conditional mutual information (CMI) construction of Steinke and Zakynthinou. In a previous work, Haghifam et al. proposed a different bound combining the two aforementioned techniques, which we refer to as the conditional individual mutual information (CIMI) bound. However, in a simple Gaussian setting, both the CMI and the CIMI bounds are order-wise worse than that by Bu et al.. This observation motivated us to propose the new bound, which overcomes this issue by reducing the conditioning terms in the conditional mutual information. In the process of establishing this bound, a conditional decoupling lemma is established, which also leads to a meaningful dichotomy and comparison among these information-theoretic bounds. 
\end{abstract}
\url{} 

\section{Introduction}
Bounding the generalization error of learning algorithms is of fundamental importance in statistical machine learning. The conventional approach is to bound it using a quantity related to the hypothesis class, such as the VC-dimension \cite{shalev2014understanding}, and such bounds are therefore oblivious to the learning algorithm and data distribution. The obtained results are usually rather conservative, and cannot fully explain the recent success of deep learning. Recently, information theoretic approaches that jointly take into consideration the hypothesis class, the learning algorithm, and the data distribution, has drawn considerable attention \cite{russo2016controlling, asadi2018chaining, issa2019strengthened, negrea2019information, jose2020information, wu2020information,xu2017information,bu2020tightening,steinke2020reasoning,haghifam2020sharpened}.

The effort of deriving generalization error bounds using information theoretic approaches was perhaps first initiated in \cite{russo2016controlling} and \cite{xu2017information}. The bound was further tightened in \cite{bu2020tightening}, by decomposing the error, and bounding each term individually. Steinke and Zakynthinou \cite{steinke2020reasoning} proposed a conditional mutual information (CMI) based bound, by introducing a dependence structure which resembles that in the analysis of the Rademacher complexity \cite{shalev2014understanding}. Combining the idea of error decomposition \cite{bu2020tightening} and the CMI bound in \cite{steinke2020reasoning}, Haghifam et al. \cite{haghifam2020sharpened} subsequently provided a sharpened bound based on conditional individual mutual information (CIMI).

In this work, we propose a new generalization error bound, which is also based on a combination of the error decomposition technique and the CMI construction. This new bound is motivated by the observation that in a simple Gaussian setting, the CIMI bound in \cite{haghifam2020sharpened} (as well as the CMI bound in \cite{steinke2020reasoning}) is of constant order, while the bound in \cite{bu2020tightening} is of order $\Theta(\frac{1}{\sqrt{n}})$, where $n$ is the number of training samples. We further observe that the conditioning term in CIMI is the same as CMI, and it tends to reveal too much information which makes the bounds loose. The proposed new bound is thus obtained by making the mutual information conditioned on an individual sample (pair), which we refer to as the individually conditional individual mutual information (ICIMI) bound. In order to establish the new bound, we introduce a new conditional decoupling lemma. This lemma allows us to view the bounds in \cite{xu2017information,bu2020tightening,steinke2020reasoning,haghifam2020sharpened} and the new bound in a unified manner, which not only yields a dichotomy of these bounds, but also makes possible a meaningful comparison among them. Finally, we show that in the Gaussian setting mentioned earlier, the proposed new bound is also able to provide a bound of the same order as, but with an improved leading constant than, that in \cite{bu2020tightening}. 

After our initial preprint was posted on Arxiv, we were made aware of an independent work by Rodr\'{i}guez-G\'{a}lvez et al. \cite{rodriguez2020random}, where a similar ICIMI-based generalization bound was proposed under the  restricted assumption of bounded loss. In contrast, our result applies under more general conditions. Our work was mainly motivated by the looseness of the CIMI bound in the Gaussian setting, for which the restricted assumption in \cite{rodriguez2020random} makes their result not applicable. Furthermore, the proposed conditional decoupling lemma, which we believe is of fundamental importance, was not present in \cite{rodriguez2020random}.

\section{Preliminary}
We study the classic supervised learning setting. Denote the data domain as $\Zc := \Xc \times \Yc$, where $\Xc$ is the feature domain and $\Yc$ is the label set. The parametric hypothesis class is denoted as $\Hc_\Wc = \{ h_{W} : W \in \Wc \} \subseteq \Yc^{\Xc}$, where $\Wc$ is the parameter space. During the training, the learning algorithm (learner) has access to a sequence of training samples  $Z_{[n]} = (Z_1, Z_2, \ldots, Z_n)$, where each $Z_i$ is drawn independently from $\Zc$ following some unknown probability distribution $\xi$. The learner can be represented by $P_{W | Z_{[n]}}$, which is a kernel (channel) that (randomly) maps $\Zc^n$ to $\Wc$. 

To complete the classification or regression task, the learner in principle would choose a hypothesis $w \in \Wc$ to minimize the following population loss, under a given loss function $\ell:  \Wc  \times \Zc \rightarrow \Rb$, 
\begin{align}
L_{\xi}(w) = \int_{\Zc} \ell(w, z) \xi(d z). \label{eqn:def-pop-loss}
\end{align}
However, since only a training data vector $Z_{[n]}$ is available, the empirical loss of $w$ is usually computed (and minimized during training), which is given as 
\begin{align}
L_{Z_{[n]}}(w) = \frac{1}{n} \sum_{i = 1}^n \ell(w, Z_i). \label{eqn:def-emp-loss}
\end{align}
The expected generalization error of the learner $P_{W|Z_{[n]}}$ is
\begin{align}
\gen(\xi, P_{W|Z_{[n]}}) := \Eb \left[L_{\xi}(W) - L_{Z_{[n]}}(W) \right], \label{eqn:def-gen}
\end{align}
where the expectation is taken over the joint distribution $P(W,Z_{[n]})=\xi^n \otimes P_{W|Z_{[n]}}$. This quantity captures the effect of the learner's expected overfitting error due to limited training data, which we shall study in this work.

\section{Review of Related Results}
In this section, we briefly review a few information theoretic bounds on the generalization error relevant to this work. A more thorough discussion of their relation is deferred to Section \ref{sec:dichotomy} and \ref{sec:comparison}, after a unified framework is given. 

\subsection{Mutual information based bounds}

Xu and Raginsky, motivated by a previous work by Russo and Zou \cite{russo2016controlling}, provided a mutual information (MI) based bound on the expected generalization error \cite{xu2017information}. 
\begin{theorem}[MI Bound \cite{xu2017information}] \label{thm:MI} 
Suppose $\ell(w, Z)$ is $\sigma^2$-sub-Gaussian under $\xi$ for all $w \in \Wc$, then
\begin{align}
\gen( \xi, P_{W|Z_{[n]}} ) \leq \sqrt{\frac{2 \sigma^2}{n} I\left(W; Z_{[n]} \right)}. \label{eqn:MI}
\end{align}
\end{theorem}

The generalization can be written in two ways
\begin{align}
\gen( \xi, P_{W | Z_{[n]}} ) &= \Eb\left[L_{\tilde{Z}_{[n]}}(\tilde{W}) \right] - \Eb\left[L_{Z_{[n]}}(W) \right]\label{eqn:gen-decouple}\\
&= \frac{1}{n} \sum_{i = 1}^n \Eb\left[ (\ell(\tilde{W}, \tilde{Z}_i) - \ell(W, Z_i)) \right], \label{eqn:gen-ind}
\end{align}
where $\tilde{W}$ and $\tilde{Z}_i$ are independent random variables that have the same marginal distributions as $W$ and $Z_i$, respectively. Instead of bounding the difference (\ref{eqn:gen-decouple}) as in \cite{xu2017information}, Bu et al. \cite{bu2020tightening} bounded each individual difference in (\ref{eqn:gen-ind}) and derived an individual mutual information (IMI) based bound. Furthermore, the following inverse Fenchel conjugate function was utilized to obtain a tightened bound. For any random variables $F$, its cumulant generating function is 
\begin{align}
\psi_F(\lambda) := \ln \Eb\left[ e^{ \lambda (F - \Eb[F ])} \right], 
\end{align}
and the inverse of its Fenchel conjugate is given as
\begin{align}
\psi^{*-1}_{F}(\eta) := \inf_{ \lambda > 0} \frac{\eta + \psi_{F}(\lambda) }{\lambda}, \quad \eta \in [0, \infty). \label{eqn:def-origin-conjate-inverse}
\end{align}
The tightened bound is summarized in the following theorem. 
\begin{theorem}[IMI Bound \cite{bu2020tightening}] \label{thm:IMI} 
Suppose $\psi_{-}$ is an upper bound of $\psi_{-\ell(\tilde{W}, \tilde{Z}_i)}$, then
\begin{align}
\gen( \xi, P_{W|Z_{[n]}} ) \leq \frac{1}{n} \sum_{i = 1}^n \psi^{*-1}_{-}\left(I\left(W; Z_i \right) \right),
\end{align}
where $\tilde{W}$ and $\tilde{Z}_i$ are independent random variables that have the same marginal distributions as $W$ and $Z_i$, respectively.
\end{theorem}

\subsection{Conditional mutual information based bounds}\label{sec:CMI}
Steinke and Zakynthinou \cite{steinke2020reasoning} recently introduced a novel bounding approach. In their approach, $Z_{[n]}^\pm := (Z^{\pm1}_{1}, Z^{\pm1}_{2}, \ldots, Z^{\pm1}_{n})$ is a $2 \times n$ table of samples that each $Z_{i}^{s}$, for $s = -1,1$ and $i = 1,\ldots, n$ is independently drawn following $\xi$. The training vector $(Z^{R_1}_{1}, Z^{R_2}_{2}, \ldots, Z^{R_n}_{n})$ is selected from the table $Z^{\pm}_{[n]}$, where $R_i$'s are independent Rademacher random variables, i.e., $R_i$ takes $1$ or $-1$ equally likely. The vector $R_{[n]} = (R_1, \ldots, R_n) \in \{-1, 1\}^n$ essentially selects one sample from each column in the table, which partition $Z_{[n]}^{\pm}$ into a training vector and a testing vector. For simplicity, we shall write $Z_i^{-1}$ and $Z_{i}^{+1}$ as $Z_i^{-}$ and $Z_{i}^{+}$, when the meaning is clear from the context. 

With the structure given above, the expected generalization error of the algorithm can be written as 
\begin{align}
&\text{gen}(\xi, P_{W|Z_{[n]}}) = \notag \\
& \Eb_{Z^{\pm}_{[n]}}\left[ \Eb \left[ \frac{1}{n} \sum_{i = 1}^n R_i\left( \ell(W, Z_{i}^{-}) - \ell(W, Z_{i}^{+}) \right) {\Big |} Z_{[n]}^{\pm} \right] \right]. \label{eqn:gen-cond}
\end{align}
Steinke and Zakynthinou obtained the following conditional mutual information (CMI) based result. 
\begin{theorem}[CMI Bound \cite{steinke2020reasoning}] \label{thm:CMI}
Suppose $\sup_{w \in \Wc} |\ell(w, z_1) - \ell(w, z_2)| \leq \Delta(z_1, z_2)$ for any $z_1, z_2 \in \Zc$, then
\begin{align}
\gen( \xi, P_{W|Z_{[n]}} ) \leq \sqrt{\frac{2}{n}\Eb[\Delta(Z_1, Z_2)^2] I\left(W; R_{[n]} | Z^{\pm}_{[n]}\right)}, \label{eqn:CMI-zz}
\end{align}
where $Z_1, Z_2$ are independent samples distributed as $\xi$.
\end{theorem}
Since $R_i$ is binary, the conditional mutual information is always bounded; in contrast, mutual information based bounds (i.e., MI and IMI bounds) can be unbounded, particularly when the random variables $W, Z_{i}$ are both continuous. 

Motivated by the results in \cite{bu2020tightening}, Haghifam et al. \cite{haghifam2020sharpened} proposed a sharpened bound by similarly bounding each term in (\ref{eqn:gen-cond}). Moreover, they provided a conditional individual mutual information (CIMI) based bound represented by \textit{the sample-conditioned mutual information}, which is defined as 
\begin{align}
I_{u}(X; Y) := I(X; Y|U=u).
\end{align}
Clearly $I_U(X; Y)$ is a function of the random variable $U$, thus also a random variable, and $\Eb[I_U(X; Y)] = I(X; Y | U)$. These sharpened bounds are summarized in the following theorem. 
\begin{theorem}[CIMI Bound \cite{haghifam2020sharpened}] \label{thm:CIMI}
Suppose $\ell \in [0, 1]$, then
\begin{align}
\gen( \xi, P_{W|Z_{[n]}} ) &\leq \frac{1}{n} \sum_{i = 1}^n \Eb \left[ \sqrt{2 I_{Z^{\pm}_{[n]}}\left(W; R_i \right)} \right] \label{eqn:CIMI-hat} \\
&\leq \frac{1}{n}\sum_{i = 1}^n \sqrt{2I\left(W; R_i | Z^{\pm}_{[n]} \right)} \label{eqn:CIMI}.
\end{align}
\end{theorem}

\section{New Result}\label{sec:new}
\subsection{A motivating example} \label{sec:motivating-example}
Let us consider the simple setting of estimating the mean from samples generated from a Gaussian distribution $N(\mu, \sigma^2)$, by averaging the $i.i.d.$ training samples under the squared loss. 
\begin{example}[Estimating the Gaussian mean]\label{eg:Gaussian-mean}
The training samples $Z_{[n]}$ are drawn $i.i.d.$ following $N(\mu, \sigma^2)$ for some unknown $\mu$. The learner deterministically estimates $\mu$ by averaging the training samples, i.e., $W = \frac{1}{n}\sum_{i = 1}^n Z_i$, whose empirical error is
\begin{align}
L_{Z_{[n]}}(W) = \frac{1}{n}\sum_{i = 1}^n (W - Z_{i})^2.
\end{align}
\end{example}

Bu et al. \cite{bu2020tightening} showed that the mutual information term in the IMI bound is
\begin{align}
I(W;Z_i)=\frac{1}{2}\log\frac{n}{n-1}=\frac{1}{2(n-1)}+o\left( \frac{1}{n} \right),\label{eqn:IMI-Gaussian}
\end{align}
and obtained the following IMI based bound
\begin{align}
\sigma^2 \sqrt{\frac{2(n+1)^2}{n^2} \log\frac{n}{n-1}} = \sigma^2\sqrt{\frac{2}{n-1}} + o\left( \frac{1}{\sqrt{n}} \right). \label{eqn:bu-Gaussian}
\end{align}
For this simple setting, the generalization error can in fact be calculated exactly to be $\frac{2 \sigma^2}{n}$. Though the error bound above does not have the same order as the true generalization error, it is consistent with the VC dimension-based bound and is the best known for this case. Note that the MI bound will be unbounded, since $I(W; Z_{[n]})$ is unbounded.

Next consider the CMI and CIMI  bounds, and let us focus on the mutual information terms in these bounds, which give
\begin{align}
&I(W; R_{[n]}|Z^{\pm}_{[n]})=n/\log_2 e, \\
&I_{Z^\pm_{[n]}}(W; R_i)=1/\log_2 e, \quad a.s..
\end{align}
It is seen that they are order-wise worse than (\ref{eqn:IMI-Gaussian}), which suggests that the bounds obtained from the CMI and CIMI bounds would be order-wise worse than (\ref{eqn:bu-Gaussian}). 

Theorem \ref{thm:CMI} and Theorem \ref{thm:CIMI} in fact do not apply directly in this setting, since their required conditions do not hold. In Theorem \ref{thm:CMI}, the function $\Delta(z_1, z_2)$ does not exist (i.e., unbounded); even if it existed, the term $\Eb[\Delta(Z_1, Z_2)^2]$ would be a constant, thus the CMI bound would be of constant order. Similarly, if the condition $\ell\in[0,1]$ held, the CIIMI bound would also be of constant order. As we shall show shortly, the CMI and CIMI bounds can be generalized and strengthened, yet the resultant strengthened bounds in this setting still do not diminish as $n\rightarrow \infty$, and thus would be order-wise worse than the IMI bound. 

A question arises naturally: Is the looseness of the CMI and CIMI bounds here due to the introduction of the conditioning terms? As we shall show next, it is in fact caused by too much information being revealed in the conditioning terms, and there is indeed a natural way to resolve this issue.

\subsection{A conditional decoupling lemma}
Our main result relies on a key lemma. A few more definitions are first introduced in order to present this lemma and the main result. 

For any random variables $F$ and $U$, define the \textit{sample-conditioned cumulant generating function} for any realization $U = u$,
\begin{align}
\psi_{F|U}(\lambda, u) := \ln \Eb\left[ e^{ \lambda (F - \Eb[F | U = u])} {\Big |} U = u\right], \quad \lambda \in \Rb.
\end{align}
It is straightforward to verify that for any realization $U=u$, $\psi_{F|U}(0, u) = \psi'_{F|U}(0, u)= 0$ and $\psi_{F|U}''(0, u) > 0$. Hence the inverse of its Fenchel conjugate
\begin{align}
\psi^{*-1}_{F|U}(\eta, u) := \inf_{ \lambda > 0} \frac{\eta + \psi_{F|U}(\lambda, u) }{\lambda}, \quad \eta \in [0, \infty) \label{eqn:def-conjate-inverse}
\end{align}
is concave and non-decreasing; see e.g., \cite{bu2020tightening} and \cite{boucheron2013concentration}. The unconditioned version of this function was introduced earlier by Bu et al. \cite{bu2020tightening}. When it is clear from context, we will write 
\begin{align}
\Psi_{F|U}(\lambda) := \psi_{F|U}(\lambda, U), \quad \Psi^{*-1}_{F|U}(\eta) := \psi^{*-1}_{F|U}(\eta, U),
\end{align}
which are functions of $U$, thus random. Next define \textit{the conditional cumulant generating function} 
\begin{align}
\bar{\psi}_{F|U} = \Eb\left[\Psi_{F|U}\right],
\end{align}
 and similarly its inverse Fenchel conjugate as $\bar{\psi}_{F|U}^{*-1}$. 

For a pair of random variables $(X, Y)$, its \textit{decoupled pair conditioned on a third random variable} $U$ is a pair of random variables  $(\tilde{X}, \tilde{Y})$ , such that
\begin{align}
(\tilde{X}, U) \overset{D}{=} (X, U), \quad (\tilde{Y}, U) \overset{D}{=} (Y, U),
\end{align}
i.e., $(\tilde{X}, U)$ and $(X, U)$ are identically distributed, and $(\tilde{Y}, U)$ and $(Y, U)$ are identically distributed, 
and moreover 
\begin{align}
\tilde{X} \leftrightarrow U \leftrightarrow \tilde{Y}
\end{align}
forms a Markov string. It follows from this definition that 
\begin{align}
I_U(X ; Y)= D(P_{X,Y|U} || P_{\tilde{X},\tilde{Y}|U}). \quad\label{eqn:IU}
\end{align}

We next introduce a conditional decoupling (CD) lemma, which serves an instrumental role in our work. The unconditioned version was presented in \cite{bu2020tightening}.  
\begin{lemma}[The CD lemma] \label{lem:CD}
For any three random variables $X, Y, U$, let $\tilde{X}, \tilde{Y}$ be the decoupled pair of $X, Y$ conditioned on $U$.
Let $F := f(X, Y)$ and $\tilde{F}:= f(\tilde{X}, \tilde{Y})$, for some real-valued measurable function $f$. The following inequalities hold
\begin{align}
\Eb[F] - \Eb[ \tilde{F}] &\leq \Eb\left[ \Psi^{*-1}_{\tilde{F}|U}\left( I_{U}(X; Y) \right) \right] \notag \\
&\leq \bar{\psi}^{*-1}_{\tilde{F}|U} \left( I(X; Y|U) \right), \label{eqn:CD+}\\
\Eb[\tilde{F}] - \Eb[F] &\leq \Eb \left[ \Psi^{*-1}_{-\tilde{F}|U}\left( I_{U}(X; Y) \right) \right] \notag \\
&\leq \bar{\psi}^{*-1}_{-\tilde{F}|U} \left( I(X; Y|U) \right). \label{eqn:CD-}
\end{align}
\end{lemma}

This lemma is proved by utilizing the Donsker–Varadhan variational representation of KL divergence and the concavity of the inverse Fenchel conjugate function. The proof details are deferred to Section \ref{sec:CD}.

\subsection{The ICIMI bound}\label{sec:results}
Let $(W,Z^{\pm}_{[n]},R_{[n]})$ be as given previously in Section \ref{sec:CMI}. For each $i= 1,\ldots, n$, let $(\tilde{W}_i, \tilde{R}_i)$ be a decoupled pair of $(W, R_i)$ conditioned on $Z_{i}^\pm$. The new bound we propose is presented in Theorem \ref{thm:ICIMI}. 
\begin{theorem}\label{thm:ICIMI}(ICIMI Bound)
Given an algorithm $P_{W|Z_{[n]}}$, the following bounds on the generalization hold
\begin{align}
\gen(\xi, P_{W|Z_{[n]}}) &\leq \frac{1}{n}\sum_{i = 1}^n \Eb\left[ \Psi^{*-1}_{\tilde{G}_i|Z^{\pm}_i}(I_{Z^\pm_i}(W; R_i)) \right]\\
&\leq \frac{1}{n}\sum_{i = 1}^n \bar{\psi}^{*-1}_{\tilde{G}_i |Z^{\pm}_i}(I(W; R_i| Z^\pm_i)),
\end{align}
where $\tilde{G}_i = \tilde{R}_i\left( \ell(\tilde{W}_i, Z_{i}^{-}) - \ell(\tilde{W}_i, Z_{i}^{+}) \right)$.
\end{theorem}

There are two bounds in this theorem. The stronger bound is in terms of the sample-conditioned mutual information, which is different from the conventional notion of conditional mutual information and may be more difficult to evaluate. The weaker bound is in terms of the conventional mutual information. 

In the proposed bounds, the mutual information is conditioned on the individual data pair $Z^{\pm}_i$, instead of the full data pair set $Z^{\pm}_{[n]}$. Intuitively, revealing only $Z^{\pm}_i$ makes it more difficult,  than revealing all data pairs $Z^{\pm}_{[n]}$, to deduce information regarding $R_i$ from $W$. As a consequence, the mutual information $I(W; R_i| Z^\pm_i)$ is less than $I(W; R_i| Z^\pm_{[n]})$, yielding a potentially tighter bound.

\begin{proof}[Proof of Theorem \ref{thm:ICIMI}]
We can rewrite the generalization error given in (\ref{eqn:gen-cond}) as
\begin{align}
&\gen(\xi, P_{W|Z_{[n]}}) =   \notag\\
&\quad \frac{1}{n} \sum_{i = 1}^n \Eb \left[ \Eb \left[ R_i \left( \ell(W, Z^{-}_{i}) - \ell(W, Z^{+}_{i}) \right)| Z^{\pm}_{i} \right] \right]. \label{eqn:gen-ind-cond} 
\end{align}
Now apply the CD lemma on each individual term in (\ref{eqn:gen-ind-cond}) by letting $X=W$, $Y_i =R_i$, $U_i =Z^{\pm}_i$, and $F_i = R_i \left( \ell(W, Z^{-}_{i}) - \ell(W, Z^{+}_{i}) \right)$. Since 
\begin{align}
\Eb[\tilde{G}_i] = \Eb[\tilde{F}_i] = \Eb\left[ \tilde{R}_i\left( \ell(\tilde{W}_i, Z_{i}^{-}) - \ell(\tilde{W}_i, Z_{i}^{+}) \right) \right] = 0, \notag
\end{align}
we have 
\begin{align}
 \gen( \xi, P_{W|Z_{[n]}}) &= \frac{1}{n} \sum_{i = 1}^n \Eb[F_i] = \frac{1}{n} \sum_{i = 1}^n \Eb[F_i] - \Eb[\tilde{F}_i] \notag \\
&\leq \frac{1}{n}\sum_{i = 1}^n \Eb\left[ \Psi^{*-1}_{\tilde{G}_i|Z^{\pm}_i}(I_{Z^\pm_i}(W; R_i)) \right]\\
& \leq \frac{1}{n}\sum_{i = 1}^n \bar{\psi}^{*-1}_{\tilde{G}_i |Z^{\pm}_i}(I(W; R_i| Z^\pm_i)),
\end{align}
which completes the proof.
\end{proof}

We call this bound the individually conditional individual mutual information (ICIMI) bound, since it is derived by applying the CD lemma on the individual conditional terms in (\ref{eqn:gen-ind-cond}). 

We note that Theorem \ref{thm:ICIMI} implies Proposition 3 in \cite{rodriguez2020random}, which we state below as a corollary.
\begin{corollary}\label{cor:ICIMI-01}
Suppose $\ell \in [a, b]$ with $a < b$, then
\begin{align}
\gen( \xi, P_{W|Z_{[n]}})  &\leq \frac{b-a}{n}\sum_{i = 1}^n \Eb_{Z^{\pm}_{[n]}}\left[  \sqrt{2 I_{Z^{\pm}_{i}}(W; R_i)} \right]\\
 &\leq \frac{b-a}{n}\sum_{i = 1}^n \sqrt{2 I(W; R_i | Z^{\pm}_{i})}. \label{eqn:ICIMI}
\end{align}
\end{corollary}
\begin{proof}[Proof of Corollary \ref{cor:ICIMI-01}]
When $\ell \in [a, b]$ and $\tilde{F}_i \in [a-b, b-a]$, it is straightforward to verify that $\tilde{F}_i$ is $\frac{(b-a)^2}{2}$-sub-Gaussian. The definition of the sub-Gaussian distribution in fact gives $\Psi_{\tilde{F}_i|Z^{\pm}_i}(\lambda) \leq \frac{(b-a)^2}{2} \lambda^2$, and
thus $\Psi^{*-1}_{\tilde{F}_i|Z^{\pm}_i}(\eta) \leq (b-a)\sqrt{2 \eta}$, from which the corollary follows.
\end{proof}

\subsection{Dichotomy and generalizations of existing bounds}\label{sec:dichotomy}
The CD lemma allows us to view the existing MI, IMI, CMI, and CIMI bounds in a unified framework. By applying the CD lemma in different manners, these bounds can be obtained almost directly. The technical conditions under which the bound hold can also be generalized, and the bounds themselves can be strengthened using the inverse Fenchel conjugate. These results are summarized in Table \ref{tbl:gen-loss}. We also provide the bounds for bounded loss function, which eliminate the $\bar{\psi}^{*-1}$ functions.

\begin{table*}[tbh]
\centering
\caption{A dichotomy of several generalization bounds using the CD Lemma} \label{tbl:gen-loss}
\resizebox{0.98\textwidth}{!}{
 \begin{tabular}{||c | c | c | c | c | c | c ||} 
 \hline
 Approach &  $X$ & $Y$ & $U$ & $F$ & Generalization bound & Special case $\ell \in [0,1]$\\
 \hline\hline
 MI \cite{xu2017information} & $W$ & $Z_{[n]}$ &   & $\frac{1}{n}\sum_{i = 1}^n \ell(W, Z_{i})$ & $\bar\psi^{*-1}_{-\tilde{F}} \left(I\left(W; Z_{[n]} \right) \right)$ & $\sqrt{\frac{1}{2n}I(W; Z_{[n]})}$\\
 \hline
 IMI \cite{bu2020tightening} & $W$ & $Z_{i}$ &  & $F_i =  \ell(\tilde{W}, Z_{R_i, i})$  & $\frac{1}{n}\sum_{i = 1}^n \bar\psi^{*-1}_{-\tilde{F}_i}\left(I\left(W; Z_{i} \right) \right)$ & $\frac{1}{n}\sum_{i = 1}^n\sqrt{\frac{1}{2}I(W; Z_{i})}$ \\
 \hline
 CMI \cite{steinke2020reasoning} & $W$ & $R_{[n]}$ & $Z_{[n]}^{\pm}$ & $ \frac{1}{n} \sum_{i = 1}^n R_i\left( \ell(W, Z_{i}^{-}) - \ell(W, Z_{i}^{+}) \right)$  & $\bar\psi^{*-1}_{\tilde{F} | Z_{[n]}^{\pm}} \left( I\left(W; R_{[n]} | Z_{[n]}^{\pm}  \right) \right)$ & $\sqrt{2 I(W; R_{[n]} | Z^\pm_{[n]})}$\\
 \hline
 CIMI \cite{haghifam2020sharpened} & $W$ & $R_{i}$ & $Z_{[n]}^{\pm}$ & $F_i = R_i\left( \ell(W, Z_{i}^{-}) - \ell(W, Z_{i}^{+}) \right)$ & $\frac{1}{n}\sum_{i = 1}^n \bar\psi^{*-1}_{\tilde{F}_i | Z_{[n]}^{\pm}} \left( I\left(W; R_i | Z_{[n]}^{\pm} \right) \right)$ & $\frac{1}{n}\sum_{i = 1}^n \sqrt{2 I(W; R | Z_{[n]}^{\pm})}$\\
 \hline
 ICIMI (new) & $W$ & $R_{i}$ & $Z_{i}^{\pm}$ & $ F_i = R_i\left( \ell(W, Z_{i}^{-}) - \ell(W, Z_{i}^{+}) \right)$  & $\frac{1}{n}\sum_{i = 1}^n \bar\psi^{*-1}_{\tilde{F}_i | Z_{i}^{\pm}} \left( I\left(W; R_i | Z_{i}^{\pm} \right) \right)$ & $\frac{1}{n}\sum_{i = 1}^n \sqrt{2 I(W; R | Z_{i}^{\pm})}$\\
 \hline
\end{tabular}
}
\end{table*}

The CMI and CIMI results can be further strengthened by utilizing the inverse Fenchel conjugate function together with the sample-conditioned mutual information. More precisely, let $(\tilde{R}_{[n]}, \tilde{W})$ be the decoupled pair of $(R_{[n]}, W)$ conditioned on $Z^\pm_{[n]}$. Further define
\begin{align}
\tilde{E}_i = \tilde{R}_i\left( \ell(\tilde{W}, Z_{i}^{-}) - \ell(\tilde{W}, Z_{i}^{+}) \right),\quad \tilde{E} = \frac{1}{n} \sum_{i =1}^n \tilde{E}_i,
\end{align}
then we have the strengthened CMI and CIMI bounds:
\begin{align}
\gen\left( \xi, P_{W|Z_{[n]}} \right) &\leq \Eb \left[ \Psi^{*-1}_{\tilde{E} | Z^{\pm}_{[n]}} \left(I_{Z^{\pm}_{[n]}}\left(W; R_{[n]} \right) \right)\right], \label{eqn:CMI-general}\\
\gen\left( \xi, P_{W|Z_{[n]}} \right) &\leq \frac{1}{n}\sum_{i = 1}^n \Eb\left[ \Psi^{*-1}_{\tilde{E}_i |Z^{\pm}_{[n]}}(I_{Z^\pm_{[n]}}(W; R_i)) \right]. \label{eqn:CIMI-general}
\end{align}

\subsection{Comparison of the bounds}\label{sec:comparison}
We first consider the special case where the loss function is bounded, i.e., $\ell \in [0, 1]$. For this case, it was shown in \cite{haghifam2020sharpened} that the CIMI bound (\ref{eqn:CIMI}) is tighter than the CMI bound (\ref{eqn:CMI-zz}). We next show that the proposed bound (\ref{eqn:ICIMI}) is tighter than the CIMI bound (\ref{eqn:CIMI}) when $\ell \in [0, 1]$.

\begin{lemma}\label{lem:ind-better}
For any $i = 1,\ldots, n$, we have
\begin{align}
I(W; R_i | Z^{\pm}_i) \leq I(W; R_i | Z^{\pm}_{[n]}).\notag
\end{align}
\end{lemma}
\begin{proof}[Proof of Lemma \ref{lem:ind-better}]
By the independence of $R_i$ and $Z^{\pm}_{[n]}$, we have
\begin{align}
I(W; R_i | Z^{\pm}_{[n]}) &= H(R_i) - H(R_i | W, Z^{\pm}_{[n]}), \notag \\
I(W; R_i | Z^{\pm}_i) &= H(R_i) - H(R_i | W, Z^{\pm}_i). \notag
\end{align}
It follows that
\begin{align}
&I(W; R_i | Z^{\pm}_{[n]}) - I(W; R_i | Z^{\pm}_{i})=I(R_i; Z^{\pm}_{[n]} | W, Z^{\pm}_i) \geq 0,\notag
\end{align}
which concludes the proof.
\end{proof}

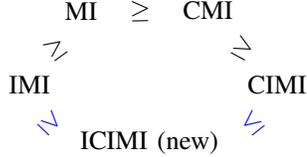
\begin{figure}[t!]
\centering
\begin{tikzpicture}[
roundnode/.style={circle, draw=green!60, fill=green!5, very thick, minimum size=7mm},
squarednode/.style={rectangle, draw=red!60, fill=red!5, very thick, minimum size=5mm},
]
\node[]      (ICIMI)                              {ICIMI (new)};
\node[] (IMI) [above left = 0.25cm and 0.15cm of ICIMI] {IMI};
\node[] (CIMI) [above right = 0.25cm and 0.15cm of ICIMI] {CIMI};
\node[] (MI) [above right =  0.5cm and -0.05cm of IMI] {MI};
\node[] (CMI) [above left = 0.5cm and -0.05cm of CIMI] {CMI};
\draw[-, draw=none] (MI.east) -- (CMI.west) node[midway, sloped] (TextNode) {$\geq$};
\draw[-, draw=none] (MI.south) -- (IMI.north) node[midway, sloped] (TextNode) {$\geq$};
\draw[-, draw=none] (IMI.south) -- (ICIMI.west) node[midway, sloped] (TextNode) {\textcolor{blue}{$\geq$}};
\draw[-, draw=none] (CMI.south) -- (CIMI.north) node[midway, sloped] (TextNode) {$\geq$};
\draw[-, draw=none] (CIMI.south) -- (ICIMI.east) node[midway, sloped] (TextNode) {\textcolor{blue}{$\leq$}};
\end{tikzpicture}
\caption{Relations among generalization bounds, when the inverse Fenchel conjugate functions are assumed to be the same. \label{fig:cmp}}
\end{figure}

To further understand the relation among these bounds under more general conditions when the loss function may not be bounded, let us assume the inverse Fenchel conjugate functions, which roughly capture the geometry induced by the expected loss, are the same (denoted as $\bar{\psi}^{*-1}$) for all the five approaches, i.e., 
\begin{align}
\bar{\psi}^{*-1}= \bar\psi^{*-1}_{-\tilde{F}} =  \bar\psi^{*-1}_{-\tilde{F}_i} = \bar\psi^{*-1}_{\tilde{F} | Z_{[n]}^{\pm}} = \bar\psi^{*-1}_{\tilde{F}_i | Z_{[n]}^{\pm}}  =  \bar\psi^{*-1}_{\tilde{F}_i | Z_{i}^{\pm}}. \notag
\end{align}
Then we can focus on the information measure quantities, and compare these bounds as shown in Fig. \ref{fig:cmp}. Here the inequalities given in black were proved previously (see \cite{bu2020tightening} and \cite{haghifam2020sharpened}). Since the common function $\bar{\psi}^{*-1}$ is non-decreasing, the inequality "CIMI \textcolor{blue}{$\geq$} ICIMI" follows from Lemma \ref{lem:ind-better}. The inequality "IMI \textcolor{blue}{$\geq$} ICIMI" is implied by the following lemma for the same reason. 
\begin{lemma}\label{lem:CMI-ICIMI}
For any $i = 1,\ldots, n$, we have
\begin{align}
I(W; R_i | Z^{\pm}_i) \leq I(W; Z_i).\notag
\end{align}
\end{lemma}
\begin{proof}[Proof of Lemma \ref{lem:CMI-ICIMI}]
First $Z_i$ and $Z_i^{R_i}$ are both the $i^{th}$ training sample for the input of the algorithm, thus
\begin{align}
& I(W; Z_i) = I(W; Z_{i}^{R_i}).
\end{align}
Then since $Z_i^{-R_i}, R_i$ and $W$ are independent given $Z_i^{R_i}$, 
\begin{align}
& I(W; Z_{i}^{\pm}, R_i) = I(W; Z_i^{R_i}, Z_i^{-R_i}, R_i)\\
& = I(W; Z_i^{R_i}) + I(W; Z_{i}^{-R_i}, R_i | Z_i^{R_i}) = I(W; Z_i^{R_i}).
\end{align}
It follows that
\begin{align}
I(W; Z_i) = I(W; Z_i^{\pm}, R_i)  \geq I(W; R_i | Z_i^{\pm}),
\end{align}
which concludes the proof.
\end{proof}

The inverse Fenchel conjugate functions may indeed be different for different bounds, thus although the above comparison suggests certain dominant relations, it is not clear for any specific problem, whether any particular bound is tighter than the other. This is particularly true if we use the bounds based on the inverse Fenchel conjugate, however, even for the special case of $\ell \in [0, 1]$, the different multiplicative factors and the sum-square-root forms imply that the relation can be less clear.

\subsection{Revisiting the example}\label{sec:rev-example}
We now return to the problem of estimating the Gaussian mean, and show that the proposed ICIMI bound can provide scaling behavior similar to that of IMI, thus order-wise stronger than the CMI and CIMI bounds. In fact, the bound is also strictly better than the IMI bound given in \cite{bu2020tightening} asymptotically in this setting. 

We first formally establish, as suspected previously, that the CMI and CIMI bounds are at least of constant order for this setting, the proof of which can be found in the appendix.

\begin{proposition}\label{prop:CMI-CIMI-Gaussian} 
The strengthened CMI and CIMI bounds, i.e., (\ref{eqn:CMI-general}) and (\ref{eqn:CIMI-general}), are at least $\frac{ \sigma^2}{\pi \sqrt{\log e}}$ in the problem of estimating the Gaussian mean.
\end{proposition}

The next proposition establishes a generalization error bound based on the ICIMI bound in this setting. 
\begin{proposition}\label{prop:ICIMI-Gaussian}
For the the problem of estimating the mean of the Gaussian distribution, the ICIMI bound gives
\begin{align}
\gen\left( \xi, P_{W | Z_{[n]}}\right) \leq \frac{2 \sigma^2}{\sqrt{\pi}} \sqrt{\frac{1}{n-1}} + o\left( \frac{1}{\sqrt{n}}\right). \label{eqn:our-Gaussian}
\end{align}
\end{proposition}
\noindent\textit{Remark: }This bound scales as $\Theta(\sqrt{\frac{1}{n}})$. Compared to the IMI bound in (\ref{eqn:bu-Gaussian}), the new ICIMI based bound is asymptotically tighter by a factor of $\sqrt{\frac{\pi}{2}} \approx 1.25$.

Proposition \ref{prop:ICIMI-Gaussian} is proved by studying separately the sample-conditioned individual mutual information $I_{Z_i^{\pm}}(W; R_i)$ and the inverse Fenchel conjugate functions $\Psi^{*-1}_{\tilde{G}_i|Z^{\pm}}$. For the former, since the algorithm here is averaging the samples without any prior of the Gaussian distribution, without loss of generality, we can assume the mean of the Gaussian distribution to be $0$, i.e., $\mu = 0$. Therefore, given $Z^{\pm}_i = z_{\pm} \in \Rb^2$, $W$ is mixed-Gaussian distributed, which follows $N(\frac{z_{+}}{n}, \frac{n-1}{n^2}\sigma^2)$ when $R_i = 1$ and follows $N(\frac{z_{-}}{n}, \frac{n-1}{n^2}\sigma^2)$ when $R_i = -1$. The term $I_{Z_i^{\pm}}(W; R_i)$ is thus related to the scaling behavior of the differential entropy of a mixed Gaussian distribution, which the following lemma makes more precise. 

\begin{lemma}\label{lem:mixed-Gaussian}
Let $R$ be a Rademacher random variable and $V$ be a mixed-Gaussian random variable, such that $V \sim N(\nu, \sigma^2)$ when $R=1$, and $V \sim N(-\nu, \sigma^2)$ when $R=-1$.
We have
\begin{align}
	I(V; R) = \frac{1}{2} \frac{\nu^2}{\sigma^2} + o\left( \frac{\nu^2}{\sigma^2} \right).
\end{align}
\end{lemma}

The next lemma gives an upper bound on the inverse Fenchel conjugate functions.
\begin{lemma}\label{lem:finer}
For the problem of estimating the mean of the Gaussian distribution, and any realization of $Z^{\pm}_{i} = z_{\pm} \in \Rb^2$ with $|z_+| \not= |z_-|$,
\begin{align}
	\Psi^{*-1}_{\tilde{G}_i|Z^{\pm}_i = z_{\pm}}(\eta) \leq B_{z_{\pm}, n}(\eta) = |z_{+}^2 - z_{-}^2 |\sqrt{2\eta} + \Theta\left( \frac{1}{n} \right), \notag
\end{align}
where
\begin{align}
B_{z_{\pm}, n}(\eta)&:= |z_{+}^2 - z_{-}^2 |\sqrt{2\eta} \notag\\ 
&+ \frac{2\sigma^2(z_+ - z_-)^2}{n|z_+^2 - z_-^2|} \sqrt{2\eta} + \frac{4\max\left(z_+^2, z_-^2 \right)}{n};
\end{align}
and for $|z_+| = |z_-|$,
\begin{align}
\Psi^{*-1}_{\tilde{G}_i|Z^{\pm}_i = z_{\pm}}(\eta) \leq 4\sigma\sqrt{\frac{2 \eta}{n}}|z_+| + \frac{4\max\left(z_{+}^2, z_{-}^2 \right)}{n}.
\end{align}
\end{lemma}

The proofs of these two lemmas are relegated to the appendix. With these lemmas, Proposition \ref{prop:ICIMI-Gaussian} can be proved as follows.

\begin{proof}[Proof of Proposition \ref{prop:ICIMI-Gaussian}]
First by Lemma \ref{lem:mixed-Gaussian}, we have
\begin{align}
I_{Z^{\pm}_{i}}(W; R_i) = \frac{(Z^{-}_{i} - Z^{+}_{i})^2}{8 \sigma^2 } \frac{1}{n-1} + o\left( \frac{1}{n} \right).
\end{align}
Then Theorem \ref{thm:ICIMI} and Lemma \ref{lem:finer} imply
\begin{align}
&\gen(\xi, P_{W|Z_{[n]}}) \leq \frac{1}{n}\sum_{i = 1}^n \Eb\left[ \Psi^{*-1}_{\tilde{G}_i|Z^{\pm}_i}\left( I_{Z_{\pm1, i}}(W; R_i)  \right) \right] \\
&\leq \frac{1}{n}\sum_{i = 1}^n \Eb\left[ \frac{(Z_i^- - Z_i^+)^2 |Z_i^- + Z_i^+|}{2 \sigma \sqrt{n-1}} + o\left( \frac{1}{\sqrt{n}} \right) \right]\\
&= \frac{2 \sigma^2}{\sqrt{\pi}} \sqrt{\frac{1}{n-1}} + o\left( \frac{1}{\sqrt{n}}\right),
\end{align}
which proves the proposition.
\end{proof}

\subsection{Proof of the CD lemma} \label{sec:CD}
\begin{proof}[Proof of Lemma \ref{lem:CD}]
The definition of sample-conditioned cumulant generating function implies that
\begin{align}
\Psi_{\tilde{F}| U}(\lambda) = \ln \Eb\left[e^{\lambda \tilde{F}} | U \right] - \Eb[\lambda \tilde{F} | U].
\end{align}
By the Donsker–Varadhan variational representation of KL divergence, for any $\lambda \in \Rb$
\begin{align}
\Eb[\lambda F | U] - \ln \Eb\left[e^{\lambda \tilde{F}} | U \right]&\leq D(P_{X,Y|U} || P_{\tilde{X},\tilde{Y}|U})\\
&=I_U(X ; Y), \label{eqn:var-kl}
\end{align}
where the equality is due to (\ref{eqn:IU}). It follows that for $\lambda > 0$
\begin{align}
\Eb[F |U ] - \Eb[\tilde{F} | U] &\leq \inf_{\lambda > 0} \frac{I_U(X; F) + \Psi_{\tilde{F}|U}(\lambda)}{\lambda} \\
&= \Psi^{*-1}_{\tilde{F}|U}\left( I_{U}(X; Y) \right).
\end{align}
Moreover
\begin{align}
 \Eb[F] - \Eb[\tilde{F}] &\leq \Eb\left[ \Psi^{*-1}_{\tilde{F}|U}\left( I_{U}(X; Y) \right) \right]  \\
& = \Eb\left[ \inf_{\lambda > 0} \frac{I_U(X; F) + \Psi_{\tilde{F}|U}(\lambda)}{\lambda} \right]\\
& \leq \inf_{\lambda > 0} \frac{I(X; F|U) + \Eb\left[\Psi_{\tilde{F}|U}(\lambda)\right] }{\lambda} \\
&= \bar{\psi}^{*-1}_{\tilde{F}|U}\left( I_{U}(X; Y) \right),
\end{align}
where the last inequality is by exchanging the order of expectation and infimum. Similarly, since
\begin{align}
\Psi_{-\tilde{F}| U}(-\lambda) = \ln \Eb\left[e^{\lambda \tilde{F}} | U \right] - \Eb[\lambda \tilde{F} | U],
\end{align}
with $\lambda < 0$, we have
\begin{align}
\Eb[\tilde{F} | U] - \Eb[F | U ] &\leq \inf_{\lambda < 0} \frac{I_U(X; F) + \Psi_{-\tilde{F}|U}(-\lambda)}{-\lambda} \\
&= \Psi^{*-1}_{-\tilde{F}|U}\left( I_{U}(X; Y) \right),
\end{align}
and
\begin{align}
\Eb[\tilde{F}] - \Eb[F] &\leq \Eb\left[ \Psi^{*-1}_{-\tilde{F}|U}\left( I_{U}(X; Y) \right) \right] \\
&\leq \bar{\psi}^{*-1}_{-\tilde{F}|U} \left( I(X; Y|U) \right),
\end{align}
which concludes the proof. 
\end{proof}

\section{Conclusion}\label{sec:conclusion}
We propose a new information theoretic generalization error bound, referred to as the ICIMI bound, based on a combination of the error decomposition technique and the conditional mutual information structure. Due to the reduced information content in the conditioning term, the proposed bound can be significantly tighter than several existing bounds. Particularly, when the loss function is bounded, it can be shown that the proposed bound is always tighter than the CMI and the CIMI bounds. A conditional decoupling lemma is provided which leads to a unified framework to study and compare these bounds, and it may be of independent interest.

\section*{Appendix}

\begin{proof}[Proof of Proposition \ref{prop:CMI-CIMI-Gaussian}]
For the special case $n=1$, i.e., there is only one training sample, the CMI based bound and CIMI based bound, i.e., (\ref{eqn:CMI-general}), (\ref{eqn:CIMI-general}), are equal. It is straightforward to verify that conditioned on $Z_1^{\pm}$, $\tilde{E} = \tilde{E}_1$ and $\tilde{E}$ takes $(Z_1^{-} - Z_1^{+})^2$ with probability $\frac{1}{2}$ and takes $-(Z_1^{-} - Z_1^{+})^2$ with probability $\frac{1}{2}$. Then we have
\begin{align}
\Psi_{\tilde{E} | Z_{[1]}^{\pm}}(\lambda) = \ln \cosh\left((Z_1^{-} - Z_1^{+})^2 \lambda \right).
\end{align}
Their inverse Fenchel conjugate functions are equal and by the lower bound of $\ln \cosh(\cdot)$ function in Lemma \ref{lem:lncosh-lower},
\begin{align}
 & \Psi_{\tilde{E} | Z^\pm_{[1]}}^{*-1}(\eta) =  \inf_{\lambda > 0} \frac{\eta + \Psi_{\tilde{E} | Z_{[1]}^{\pm}}(\lambda)}{\lambda} \\
 & \geq \inf_{\lambda > 0} \frac{\eta + \min\left( \frac{(Z_1^{-} - Z_1^{+})^2 \lambda}{2}, \frac{(Z_1^{-} - Z_1^{+})^4 \lambda^2}{4} \right)}{\lambda} \\
 & \geq \min\left(\frac{1}{2}, \sqrt{\eta} \right) (Z_1^{-} - Z_1^{+})^2.
\end{align}
Since $I_{Z^{\pm}_1}\left( W; R_1 \right) = 1/\log e, a.s.,$, we have
\begin{align}
 \Eb\left[ \Psi_{\tilde{E} | Z^\pm_{[1]}}^{*-1}\left(I_{Z^{\pm}_1}\left( W; R_1 \right) \right) \right] \geq \sigma^2 > \frac{\sigma^2}{\pi \sqrt{\log e}}.
\end{align}

For $n \geq 2$, denote the mean of $Z^\pm_{[n]}$ as $\bar{Z}$, from which we have
\begin{align}
\bar{Z} = \Eb\left[\tilde{W} | Z_{[n]}^{\pm} \right]. \label{eqn:barZ-W}
\end{align}
For each $i = 1,\ldots, n$, let $\Delta_i = \ell(\bar{Z}, Z_i^-) - \ell(\bar{Z}, Z_i^+)$. It follows that
\begin{align}
\Delta_i &= \left( Z_i^- - Z_i^+ \right) \left( Z_i^- + Z_i^+ - 2 \bar{Z}\right) \\
&= \left(1 - \frac{1}{n} \right) \left( \left( Z_i^-\right)^2 - \left( Z_i^+\right)^2 \right) \notag \\
& \quad - \frac{\sum_{j \not= i} (Z_j^- + Z_j^+)}{n} (Z_i^- - Z_i^+).
\end{align}
Thus
\begin{align}
\Eb[|\Delta_i|] &= \Eb\left[ \Eb\left[ |\Delta_i| | Z^{\pm}_i \right] \right] \geq \Eb\left[ {\big | }\Eb\left[ \Delta_i | Z^{\pm}_i \right] {\big |} \right] \notag \\
&= \left(1 - \frac{1}{n} \right) \Eb\left[ {\big |} \left( Z_i^-\right)^2 - \left( Z_i^+\right)^2 {\big |} \right] \\
& \geq \frac{1}{2} \Eb\left[ |Z_i^- - Z_i^+| \right] \Eb\left[ |Z_i^- + Z_i^+| \right] = \frac{2 \sigma^2}{\pi}, \label{eqn:LB-abs-Delta}
\end{align}
where the first inequality is by applying Jensen's inequality with respect to convex function $|\cdot|$; the last inequality is because $n\geq 2$ and $Z_i^- - Z_i^+$ and $Z_i^- + Z_i^+$ are independent. 
In addition, we can write
\begin{align}
& \Eb\left[\ell(\tilde{W}, Z_i^-) - \ell(\tilde{W}, Z_i^+)| Z_{[n]}^{\pm} \right] \notag \\
& =  \Eb\left[\left( Z_i^- - Z_i^+ \right) \left( Z_i^- + Z_i^+ - 2 \tilde{W} \right) | Z_{[n]}^{\pm} \right] \notag \\
& = \left( Z_i^- - Z_i^+ \right) \left( Z_i^- + Z_i^+ - 2 \Eb\left[ \tilde{W} | Z_{[n]}^{\pm} \right] \right) \notag \\
&= \left( Z_i^- - Z_i^+ \right) \left( Z_i^- + Z_i^+ - 2 \bar{Z}\right) = \Delta_i, \label{eqn:Delta-cond}
\end{align}
where the last equality is by the representation of $\bar{Z}$ in (\ref{eqn:barZ-W}).

We can then lower-bound the CMI based bound (\ref{eqn:CMI-general}) for this problem. The sample-conditioned cumulant generating function satisfies the bound shown in (\ref{eqn:long1}-\ref{eqn:long7}).
\begin{figure*}[t]
\begin{align}
\Psi_{\tilde{E} | Z_{[n]}^{\pm}}(\lambda) &= \ln \Eb\left[ \exp\left( \lambda \tilde{E} - \lambda \Eb[\tilde{E}] \right)  {\Big |} Z^\pm_i \right] \label{eqn:long1} \\
&=  \ln \Eb\left[ \exp \left( \frac{\lambda}{n} \sum_{i=1}^n \tilde{R}_i(\ell(\tilde{W}, Z^-_i) - \ell(\tilde{W}, Z^+_i)) \right) {\Big |} Z^\pm_{[n]} \right] \label{eqn:long2}\\
& = \ln \Eb\left[  \Eb \left[ \exp \left( \frac{\lambda}{n} \sum_{i=1}^n \tilde{R}_i (\ell(\tilde{W}, Z^-_i) - \ell(\tilde{W}, Z^+_i)) \right) {\Big |} Z_{[n]}^\pm, \tilde{R}_{[n]} \right] {\Big |} Z^\pm_{[n]} \right] \label{eqn:long3} \\
& \geq \ln \Eb\left[ \exp \left( \frac{\lambda}{n} \Eb \left[ \sum_{i=1}^n \tilde{R}_i (\ell(\tilde{W}, Z^-_i) - \ell(\tilde{W}, Z^+_i)) {\Big |} Z_{[n]}^\pm, \tilde{R}_{[n]} \right] \right) {\Big |} Z^\pm_{[n]} \right] \label{eqn:long4} \\
&=\ln \Eb\left[ \prod_{i = 1}^n \exp \left( \frac{\lambda}{n} \tilde{R}_i \Delta_i \right) {\Big |} Z^\pm_{[n]} \right] \label{eqn:long5} \\
&=  \sum_{i=1}^n \ln \Eb\left[ \exp\left( \frac{\lambda}{n} \tilde{R}_i \Delta_i \right) {\Big |} Z^\pm_{[n]} \right] \label{eqn:long6} \\
&= \sum_{i =1}^n \ln \cosh\left( \frac{\lambda}{n}\Delta_i \right) \geq \sum_{i = 1}^n \min\left(1, \frac{\lambda |\Delta_i|}{2n} \right) \frac{\lambda|\Delta_i|}{2n}.\label{eqn:long7}
\end{align}
\hrulefill
\end{figure*}
The first equality (\ref{eqn:long1}) is the definition of $\Psi_{\tilde{E} | Z_{[n]}^{\pm} }(\lambda)$; the second inequality (\ref{eqn:long2}) is by $\Eb[\tilde{E}] = 0$; the third equality (\ref{eqn:long3}) is by the total expectation; the first inequality (\ref{eqn:long4}) is by Jenson's inequality with respect to convex function $\exp(\cdot)$; the fourth equality (\ref{eqn:long5}) is by (\ref{eqn:Delta-cond}); the fifth equality (\ref{eqn:long6}) is by the independence of $\tilde{R}_{[n]}$ conditioned on $Z_{n}^{\pm}$; and the last inequality is due to Lemma \ref{lem:lncosh-lower}.
Its inverse Fenchel conjugate function can thus be lower bounded as follows.
\begin{align}
 &\Psi_{\tilde{E} | Z^\pm_{[n]}}^{*-1}(\eta) = \inf_{\lambda > 0} \frac{\eta + \Psi_{\tilde{E} | Z_{[n]}^{\pm}}(\lambda)}{\lambda} \\
  & \geq \inf_{\lambda > 0} \sum_{i = 1}^n \frac{\frac{1}{n}\eta + \min\left(1, \frac{\lambda |\Delta_i|}{2n} \right) \frac{\lambda|\Delta_i|}{2n} }{\lambda }\\
 & \geq \sum_{i = 1}^n \inf_{\lambda > 0}  \frac{\frac{1}{n}\eta + \min\left(1, \frac{\lambda |\Delta_i|}{2n} \right) \frac{\lambda|\Delta_i|}{2n} }{\lambda }\\
 &\geq \sum_{i = 1}^n \min\left( \frac{|\Delta_i|}{2n}, \frac{\sqrt{\eta}|\Delta_i|}{n^{3/2}} \right) \label{eqn:lower-bound-gCGF}.
\end{align}
Then since $I_{Z_{n}^\pm}(W; R_{[n]}) = n/\log e, a.s.$ and $\Psi_{\tilde{E} | Z^\pm_{[n]}}^{*-1}$ is non-negative, the CMI based bound satisfies
\begin{align}
& \Eb \left[ \Psi^{*-1}_{\tilde{E} |Z^{\pm}_{[n]}}(I_{Z^\pm_{[n]}}(W; R_{[n]})) \right] \\
& \geq \sum_{i = 1}^n \Eb\left[ \min\left( \frac{|\Delta_i|}{2n}, \frac{|\Delta_i|}{\sqrt{\log e}n} \right) \right]\\
& \geq \sum_{i = 1}^n \Eb\left[ \frac{|\Delta_i|}{2\sqrt{\log e} n} \right] \geq \frac{ \sigma^2}{\pi \sqrt{\log e}},
\end{align}
where the last equality is by (\ref{eqn:LB-abs-Delta}).

Similarly, we can lower-bound the CIMI based bound (\ref{eqn:CIMI-general}). The sample-conditioned cumulant generating function satisfies
\begin{align}
& \Psi_{\tilde{E}_i | Z_{[n]}^{\pm}}(\lambda) \notag \\ 
& =\ln \Eb\left[ \exp \left(\lambda \tilde{R}_i(\ell(\tilde{W}, Z^-_i) - \ell(\tilde{W}, Z^+_i)) \right) {\Big |} Z^\pm_{[n]} \right] \\
& = \ln \cosh(\lambda \Delta_i) \geq \min\left(1, \frac{|\lambda \Delta_i|}{2} \right) \frac{|\lambda \Delta_i|}{2}.
\end{align}
The inverse Fenchel conjugate functions can be lower bounded as 
\begin{align}
\Psi_{\tilde{E}_i | Z^\pm_{[n]}}^{*-1}(\eta) & \geq \inf_{\lambda > 0} \frac{\eta + \min\left(1, \frac{|\lambda \Delta_i|}{2} \right) \frac{|\lambda \Delta_i|}{2} }{\lambda} \\
& = \min\left( \inf_{\lambda > 0} \frac{\eta + \lambda \frac{ |\Delta_i|}{2} }{\lambda}, \inf_{\lambda > 0} \frac{\eta + \lambda^2 \frac{\Delta_i^2}{4} }{\lambda} \right) \\
& = \min\left( \frac{1}{2}, \sqrt{\eta} \right) |\Delta_i|.
\end{align}
Since $\Psi_{\tilde{E}_i | Z^\pm_{[n]}}^{*-1}(\eta)$ is non-negative, and $I_{Z_{n}^\pm}(W; R_i) = 1/\log e, a.s.$, the CIMI based bound satisfies
\begin{align}
&\frac{1}{n}\sum_{i = 1}^n \Eb \left[ \Psi^{*-1}_{\tilde{E}_i |Z^{\pm}_{[n]}}(I_{Z^\pm_{[n]}}(W; R_i)) \right] \notag\\
&\geq \frac{1}{2\sqrt{\log e} n} \sum_{i = 1}^n \Eb[|\Delta_i|] = \frac{ \sigma^2}{\pi \sqrt{\log e}}.
\end{align}

We can now conclude that the CMI and CIMI bounds in this setting are both at least $\frac{\sigma^2}{\pi \sqrt{\log e}}$.
\end{proof}

\begin{proof}[Proof of Lemma \ref{lem:mixed-Gaussian}]
By the representation of the differential entropy of mixed Gaussian distribution in \cite{michalowicz2008calculation}, we can write
\begin{align}
I(V; R) = h(V) - h(V | R) = \alpha^2 - I(\alpha), \label{eqn:mixed-Gaussian}
\end{align}
where $\alpha = \frac{|\nu|}{\sigma}$ and 
\begin{align}
&I(\alpha) =\frac{2}{\sqrt{2\pi}} e^{-\alpha^2/2} \int_{0}^{\infty} e^{-t^2/2} \text{cosh}(\alpha t) \ln(\text{cosh}(\alpha t)) d t. \notag
\end{align}
Since for any $x \in \Rb$, by the Taylor expansion,
\begin{align}
1 + x^2/2 \leq \text{cosh}(x) = \frac{1}{2} (e^{x} + e^{-x}) \leq e^{x^2/2},
\end{align}
it follows that for any $\alpha < 1$,
\begin{align}
& \frac{\frac{2}{\sqrt{2\pi}} \int_{0}^{\infty} e^{-t^2/2} \text{cosh}(\alpha t) \ln(\text{cosh}(\alpha t)) d t}{\alpha^2} \\
& \leq  \frac{\frac{2}{\sqrt{2\pi}} \int_{0}^{\infty} e^{-t^2/2} e^{\alpha^2t^2/2} \frac{\alpha^2 t^2}{2} d t}{\alpha^2} \\
& = \frac{1}{\sqrt{2\pi}} \int_{0}^{\infty} t^2 e^{-t^2(1-\alpha^2)/2} d t = \frac{1}{2 \sqrt{1 - \alpha^2}},
\end{align}
and take the limit of $\alpha^2 \rightarrow 0$ on both side,
\begin{align}
& \lim_{\alpha^2 \rightarrow 0} \frac{\frac{2}{\sqrt{2\pi}} \int_{0}^{\infty} e^{-t^2/2} \text{cosh}(\alpha t) \ln(\text{cosh}(\alpha t)) d t}{\alpha^2} \leq \frac{1}{2}.
\end{align}
In addition,
\begin{align}
& \frac{\frac{2}{\sqrt{2\pi}} \int_{0}^{\infty} e^{-t^2/2} \text{cosh}(\alpha t) \ln(\text{cosh}(\alpha t)) d t}{\alpha^2} \\
&\geq \frac{\frac{2}{\sqrt{2\pi}} \int_{0}^{\infty} e^{-t^2/2} \left(1 + \frac{\alpha^2t^2}{2} \right) \ln\left(1 + \frac{\alpha^2t^2}{2} \right) d t}{\alpha^2}
\end{align}
take the limit of $\alpha^2 \rightarrow 0$ on both side,
\begin{align}
&\lim_{\alpha^2 \rightarrow 0} \frac{\frac{2}{\sqrt{2\pi}} \int_{0}^{\infty} e^{-t^2/2} \text{cosh}(\alpha t) \ln(\text{cosh}(\alpha t)) d t}{\alpha^2} \\
&\geq \lim_{\alpha^2 \rightarrow 0} \frac{\frac{2}{\sqrt{2\pi}} \int_{0}^{\infty} e^{-t^2/2} \left(1 + \frac{\alpha^2t^2}{2} \right) \ln\left(1 + \frac{\alpha^2t^2}{2} \right) d t}{\alpha^2} \\
&= \frac{2}{\sqrt{2\pi}} \int_{0}^{\infty} e^{-t^2/2} \lim_{\alpha^2 \rightarrow 0} \frac{\left(1 + \frac{\alpha^2t^2}{2} \right) \ln\left(1 + \frac{\alpha^2t^2}{2} \right) }{\alpha^2}d t \\
&= \frac{1}{\sqrt{2\pi}} \int_{0}^{\infty} t^2 e^{-t^2/2} d t = \frac{1}{2},
\end{align}
where the first equality is by exchanging the limit and integral because function $\frac{(1+x)\ln(1+x)}{x}$ is monotonically increasing for $x \geq 0$. Thus the Taylor expansion of $I(\alpha)$ is
\begin{align}
 I(\alpha) = \frac{1}{2} \alpha^2 + o(\alpha^2),
\end{align}
plugging which in equation (\ref{eqn:mixed-Gaussian}) completes the proof.
\end{proof}

\begin{proof}[Proof of Lemma \ref{lem:finer}]
Given $Z_{i}^\pm = z_{\pm} \in \Zc^2$, $\tilde{W}_i$ and $W$ are identically distributed. Drop the index $i$ and write $\tilde{W}_i$ as $\tilde{W}$ for simplicity. With probability $1/2$, $\tilde{W} \sim N\left(\frac{z_{+}}{n}, \frac{n-1}{n^2} \sigma^2 \right)$, and with probability $1/2$, $\tilde{W} \sim N\left(\frac{z_{-}}{n}, \frac{n-1}{n^2} \sigma^2 \right)$. For any $\lambda > 0$,
\begin{align}
&\exp\left( \Psi_{\tilde{G}_i|Z^{\pm}_i=z_{\pm}}\left( \lambda \right) \right)\\
=&\Eb\left[\exp\left( \lambda \tilde{R}\left( \ell(\tilde{W}, z_{-}) - \ell(\tilde{W}, z_{+})\right)\right) \right] \\
=&\Eb \left[\exp\left( \lambda \tilde{R}\left( z_{-}^2 - z_{+}^2 + 2(z_+ - z_{-}) \tilde{W} \right)\right) \right] \\
=& \frac{1}{2}\Eb \left[\exp\left( 2\lambda (z_+ - z_{-}) \tilde{W} \right) \right]\exp\left( \lambda(z_{-}^2 - z_{+}^2) \right) \notag \\
+& \frac{1}{2}\Eb \left[\exp\left( 2\lambda (z_{-} - z_+) \tilde{W} \right) \right] \exp\left( -\lambda(z_{-}^2 - z_{+}^2) \right) \\
\leq &{\Bigg (} \frac{1}{2}\exp\left(2\lambda |z_+ - z_-|\frac{|z_-|}{n} + 2\lambda^2(z_+ - z_-)^2 \frac{n-1}{n^2}\sigma^2 \right) \notag \\
& +\frac{1}{2}\exp\left(2\lambda|z_+ - z_-|\frac{|z_+|}{n} + 2\lambda^2(z_+ - z_-)^2 \frac{n-1}{n^2}\sigma^2 \right) {\Bigg )} \notag \\
& \cdot \left( \frac{1}{2}\exp(\lambda(z_{-}^2 - z_{+}^2)) + \frac{1}{2} \exp(\lambda(z_{+}^2 - z_{-}^2)) \right) \\
\leq&\exp\left(2\sigma^2 \lambda^2 (z_{+} - z_{-})^2\left(\frac{1}{n} - \frac{1}{n^2}\right) \right) \notag \\
	&\cdot \exp\left(2\lambda |z_{+} - z_{-}|\frac{\max(|z_{+}|, |z_{-}|)}{n} \right) \notag\\
	&\cdot \left( \frac{1}{2}\exp(\lambda(z_{-}^2 - z_{+}^2)) + \frac{1}{2} \exp(\lambda(z_{+}^2 - z_{-}^2)) \right) \\
\leq & \exp\left(2\sigma^2 \lambda^2 (z_{+} - z_{-})^2 \frac{1}{n}  \right) \notag \\
	&\cdot \exp\left(2\lambda |z_{+} - z_{-}|\frac{\max(|z_{+}|, |z_{-}|)}{n} \right) \notag\\
	&\cdot \exp\left(\frac{\lambda^2}{2}(z_{-}^2 - z_{+}^2)^2\right),
\end{align}
where the last inequality is from $\frac{1}{2} (e^{x} + e^{-x}) \leq e^{x^2/2}$.
We have for any $\eta >0$,
\begin{align}
& \Psi^{*-1}_{\tilde{G}_i|Z^{\pm}_i}(\eta) = \inf_{\lambda > 0}\left\{\frac{1}{\lambda} \left(\eta +\Psi_{\tilde{G}_i|Z^{\pm}_i}\left( \lambda \right) \right) \right\}\\
& \leq \inf_{\lambda > 0} {\Bigg \{} \frac{1}{\lambda} \eta + \frac{\lambda}{2}\left(\left(Z_{i}^{+}\right)^2 - \left(Z_{i}^{-}\right)^2 \right)^2\\
&\quad\quad + \frac{2\sigma^2\lambda}{n}(Z_{i}^{+} - Z_{i}^{-})^2 + \frac{4\max(Z_{i}^{+}, Z_{i}^{-})^2}{n} {\Bigg \}}
\end{align}
It follows that if $|Z_i^+| \not= |Z_i^-|$, take $\lambda = \frac{\sqrt{ 2 \eta}}{|\left(Z_{i}^{+}\right)^2 - \left(Z_{i}^{-}\right)^2|}$
\begin{align}
& \Psi^{*-1}_{\tilde{G}_i|Z^{\pm}_i}(\eta) \leq B_{Z^{\pm}_{i}, n}(\eta),
\end{align}
and if $Z_i^+ = Z_i^-$, take $\lambda \rightarrow +\infty$,
\begin{align}
& \Psi^{*-1}_{\tilde{G}_i|Z^{\pm}_i}(\eta) \leq \frac{4\max\left(\left(Z_{1}^{+}\right)^2, \left(Z_{i}^{-}\right)^2 \right)}{n},
\end{align}
and if $Z_i^+ = - Z_i^- \not=0$, take $\lambda = \frac{1}{2\sigma|Z_i^+|}\sqrt{\frac{n\eta}{2}}$,
\begin{align}
& \Psi^{*-1}_{\tilde{G}_i|Z^{\pm}_i}(\eta) \leq 4\sigma\sqrt{\frac{2 \eta}{n}}|Z_i^+| + \frac{4\left(Z_{1}^{+}\right)^2}{n}.
\end{align}
\end{proof}

\begin{lemma}\label{lem:lncosh-lower}
The function $\ln \cosh(x)$ is lower bounded as
\begin{align}
\ln \cosh(x) \geq \min\left(1, \frac{|x|}{2} \right) \frac{|x|}{2}. \label{eqn:lncosh-lower}
\end{align}
\end{lemma}
\begin{proof}
When $|x| \geq 2$,
\begin{align}
\frac{1}{2} \left( e^{x} + e^{-x} \right) > \frac{1}{2} e^{|x|} = \frac{e^{|x|/2}}{2} e^{|x|/2} > e^{|x|/2}.
\end{align}
Take $\ln$ on both,
\begin{align}
\ln \cosh(x) &\geq \frac{|x|}{2}, \quad |x| \in [2, \infty).
\end{align}
When $|x| \leq 2$,

It is straightforward to verify by calculating derivatives that the function $\tanh(x) - \frac{x}{2}$ for $x \geq 0$ is increasing then decreasing. Since $\tanh(0) = 0$, $\ln \cosh(x) - \frac{x^2}{4}$, whose derivative is $\tanh(x) - \frac{x}{2}$, for $x \geq 0$ is increasing (then decreasing but is not needed here). Since $\ln \cosh(0) = 0$ and $\ln \cosh(2) - 1 > 0$, by the fact that $\ln \cosh(x) - \frac{x^2}{4}$ is even function, we know
\begin{align}
\ln \cosh(x) \geq \frac{x^2}{4}, \quad |x| \in [0, 2].
\end{align}
Then combine both results. It follows that 
\begin{align}
\ln \cosh(x) \geq \min\left(1, \frac{|x|}{2} \right) \frac{|x|}{2}. 
\end{align}
\end{proof}

\end{document}